\documentclass{article}
\usepackage[margin=2.5cm]{geometry}
\usepackage[utf8]{inputenc}

\usepackage{authblk}
\usepackage{algorithm}
\usepackage{algorithmicx}
\usepackage[noEnd=true]{algpseudocodex}
\usepackage{amsmath,amsfonts,graphicx,bm,amsthm,amssymb}
\usepackage[colorlinks=true, linkcolor=blue, citecolor=blue, urlcolor=blue]{hyperref}
\usepackage{braket}
\usepackage{balance}
\usepackage[compat=0.6]{yquant}
\useyquantlanguage{groups}
\usetikzlibrary{quotes}
\allowdisplaybreaks

\newcommand{\keywords}[1]{%
  \par\noindent\textbf{Keywords: } #1
}

\newcommand{\QPU}{\mathsf{QPU}}
\newcommand{\cc}{\mathsf{c.c.}}
\newcommand{\ham}{\operatorname{ham}}
\newcommand{\rank}{\operatorname{rank}}
\newcommand{\CPrank}{\operatorname{CPrank}}
\newcommand{\supp}{\operatorname{supp}}

\newtheorem{lemma}{Lemma}

\newtheorem{theorem}{Theorem}
\newtheorem{definition}{Definition}

\title{Scalable Multi-QPU Circuit Design for Dicke State Preparation: Optimizing Communication Complexity and Local Circuit Costs}

\author{
Ziheng Chen$^{1, 2}$,
Junhong Nie$^{3}$,
Xiaoming Sun$^{1, 2,}$\thanks{Corresponding authors. sunxiaoming@ict.ac.cn; zhujiadong2016@163.com},
Jialin Zhang$^{1, 2}$,
Jiadong Zhu$^{1, *}$ \\
{\small
$^{1}$ Institute of Computing Technology, Chinese Academy of Sciences \\
$^{2}$ University of Chinese Academy of Sciences \\
$^{3}$ Shandong University
}
}

\date{\today}

\begin{document}

\maketitle

\begin{abstract}
    Preparing large-qubit Dicke states is of broad interest in quantum computing and quantum metrology. However, the number of qubits available on a single quantum processing unit (QPU) is limited---motivating the distributed preparation of such states across multiple QPUs as a practical approach to scalability. In this article, we investigate the distributed preparation of $n$-qubit $k$-excitation Dicke states $D(n,k)$ across a general number $p$ of QPUs, presenting a distributed quantum circuit (each QPU hosting approximately $\lceil n/p \rceil$ qubits) that prepares the state with communication complexity $O(p \log k)$, circuit size $O(nk)$, and circuit depth $O\left(p^2 k + \log k \log (n/k)\right)$. To the best of our knowledge, this is the first construction to simultaneously achieve logarithmic communication complexity and polynomial circuit size and depth.
    
    We also establish a lower bound on the communication complexity of $p$-QPU distributed state preparation for a general target state. This lower bound is formulated in terms of the canonical polyadic rank (CP-rank) of a tensor associated with the target state. For the special case $p = 2$, we explicitly compute the CP-rank corresponding to the Dicke state $D(n,k)$ and derive a lower bound of $\lceil\log (k + 1)\rceil$, which shows that the communication complexity of our construction matches this fundamental limit.
\end{abstract}


\section{Introduction}

Despite the rapid progress of quantum hardware in recent decades, a fundamental limitation persists in the number of qubits that can be controlled and made to interact coherently on a single quantum processing unit (QPU) \cite{preskill2018quantum, arute2019quantum, bharti2022noisy}. Consequently, distributed quantum computing (DQC) has emerged as a natural, practically grounded solution for executing quantum tasks that demand large qubit counts \cite{denchev2008distributed,cacciapuoti2020quantum,caleffi2024distributed,barral2025review,long2022interfacing}. In DQC, quantum computing tasks (e.g., quantum state preparation, unitary operation execution) are implemented collaboratively across interconnected QPUs---devices linked via quantum or classical communication channels---thereby circumventing the qubit-capacity bottleneck of standalone devices while remaining fully compatible with hardware constraints \cite{li2025space, main2025distributed, davarzani2022hierarchical, ferrari2023modular, ghodsollahee2021connectivity}. DQC is usually modelled using distributed quantum circuits---collections of quantum circuits that support non-local two-qubit gates, which act on data qubits belonging to different QPUs \cite{main2025distributed, ghodsollahee2021connectivity, sarvaghad2021general, dadkhah2021new, wu2023entanglement, davis2023towards}. In DQC research, besides the traditional metrics used to quantify intra-processor computing costs (such as quantum circuit size and depth), one must also optimize inter-QPU communication costs---specifically, the number of non-local two-qubit gates \cite{ghodsollahee2021connectivity, chen2025circuit, zomorodi2018optimizing, bandic2023mapping, burt2024generalised}\footnote{While communication cost is typically quantified by the number of exchanged qubits in quantum communication complexity research, the count of non-local two-qubit gates serves as a more natural metric within the circuit model. Notably, these two measures are equivalent up to a constant factor---and thus equivalent in the complexity-theoretic sense.}.

Dicke states \cite{dicke1954coherence}---equal superpositions of all computational-basis states with fixed hamming weight---play an important role in several related areas of quantum information science, such as quantum networking\cite{prevedel2009experimental, chiuri2012experimental}, quantum game theory\cite{ozdemir2007necessary}, quantum tomography\cite{toth2010permutationally} and quantum metrology\cite{hyllus2012fisher, ouyang2021robust}.
Formally, for an $n$-qubit system, a Dicke state with $k$ excitations (denoted $D(n, k)$) is defined as
\begin{equation*}
    D(n, k) = \frac{1}{\sqrt{\binom{n}{k}}} \sum_{\substack{x \in \{0,1\}^n \\ \ham(x)=k}} \ket{x},
\end{equation*}
where the hamming weight $\ham(x)=k$ denotes the number of 1-valued bits in the $n$-bit string $x$.
Owing to their theoretically attractive properties and practical utility, a wealth of research has emerged in recent years focused on the efficient preparation of Dicke states---specifically, the optimization of the size and depth of the quantum circuits required to generate such states \cite{bartschi2019deterministic, bartschi2022short, yuan2025depth, yu2026efficient, li2025preparation}.


The potential applications of Dicke state preparation encompass quantum computing tasks demanding a large number of qubits; for instance, Dicke states $D(n, n / 2)$ surpass the classical precision limit of $O(1/\sqrt{n})$ for $n$ particles and attain the Heisenberg limit of $O(1/n)$ in quantum metrology, a scaling that enables far higher precision for large qubit counts \cite{hyllus2012fisher, ouyang2021robust}. For this reason, investigating low-complexity Dicke state preparation methods within the DQC paradigm is critical to advancing such large-scale quantum operations. In the DQC paradigm, two existing approaches can be used for preparing Dicke states across interconnected QPUs, though each suffers from notable complexity-related drawbacks. The first approach adapts circuit designs for general distributed state preparation to the specific case of Dicke states, while the second splits a single-QPU Dicke state preparation circuit into subcircuits that are executed across distributed QPUs. The former prioritizes minimizing the number of non-local gates, thereby lowering communication complexity. For instance, extending the general distributed state preparation framework from \cite{ambainis2003quantum} to Dicke states yields a logarithmic communication cost for the 2-QPU scenario\footnote{In fact, \cite{ambainis2003quantum} establishes that both the upper and lower bounds on communication complexity for general distributed state preparation over 2 QPUs equal the log-rank of a matrix associated with the target state. This work, however, does not address the special case of Dicke states; accordingly, the logarithmic upper and lower bound results for Dicke states are derived by the authors. For further details, see Section \ref{sec:lower}.}. However, this design abstracts away local computation costs, ignoring the exponential circuit size and depth incurred on each individual QPU. In contrast, the second approach partitions the qubit register of a single-QPU Dicke state preparation circuit into disjoint subsets, each assigned to a separate QPU \cite{ghodsollahee2021connectivity, burt2024generalised}. Since state-of-the-art single-QPU Dicke state preparation circuits achieve polynomial scaling in both circuit size and depth \cite{bartschi2019deterministic, bartschi2022short, yuan2025depth}, this qubit-partitioning strategy preserves polynomial circuit complexity on each QPU---a marked improvement over the first approach’s prohibitive exponential local complexity. That said, the method incurs a communication complexity that is at least polynomial in scale, representing an exponential degradation relative to the first approach’s logarithmic communication cost. Crucially, from the perspective of experimental quantum implementations, inter-processor communication imposes far more stringent bottlenecks than intra-processor computing operations in distributed quantum systems: communication introduces qubit decoherence, transmission latency, and hardware synchronization challenges that are far harder to mitigate than local gate noise. This is why, even though the second approach circumvents the first’s catastrophic local complexity issues, its polynomial communication cost still renders it undesirable for large-scale distributed Dicke state preparation.

In this paper, we seek to address the limitations of the two aforementioned approaches and propose a distributed Dicke state preparation circuit that simultaneously optimizes both communication overhead and local computation costs. Formally, we investigate the following problem: Given a Dicke state $D(n, k)$ and $p$ QPUs, each allocated approximately $\lceil n/p\rceil$ qubits\footnote{In addition to the $\lceil n/p\rceil$ data qubits, a modest number of ancillary qubits are permitted per QPU, but their usage is tightly restricted. This per-QPU qubit limit reflects practical hardware constraints: while relaxing it would enable state preparation using fewer QPUs, doing so violates the foundational goal of distributed quantum computing (DQC) research---addressing the qubit capacity limitations of standalone devices.}, can we construct a distributed quantum circuit spanning $p$ QPUs such that the intra-QPU computation cost---i.e., circuit size and depth---scales polynomially, while the inter-QPU communication complexity scales logarithmically?

\subsection{Our Contributions}
 
We answer the problem in the affirmative. Specifically, for a Dicke state $D(n,k)$, we obtain the following results:

\begin{itemize}
    \item \textbf{Efficient circuit design:} We propose a distributed Dicke state preparation circuit for deployment across $p$ QPUs, where each QPU is equipped with approximately $\lceil n/p\rceil$ qubits. This circuit achieves logarithmic-scaling communication complexity, alongside polynomial circuit size and depth. A detailed performance comparison is provided in Table \ref{tab:comparison}. Notably, for constant $p$, the circuit depth and size match those of the state-of-the-art non-distributed preparation.
    \item \textbf{Communication complexity lower bound:} We establish that for any target state $\ket{\psi}$, a communication complexity lower bound for preparing $\ket{\psi}$ accross $p$ QPUs is given by the logarithm of the CP-rank of a $p$-order tensor associated with $\ket{\psi}$. Unfortunately, computing the CP-rank of a tensor is NP-hard for $p \ge 3$\cite{haastad1990tensor}, and we have not yet succeeded in evaluating the CP-rank of Dicke states for this general case. In the special case $p = 2$, our CP-rank-based lower bound reduces to the rank-based lower bound derived in \cite{ambainis2003quantum}. In this 2-QPU scenario, we explicitly compute the lower bound, which equals $\log k$---a result that shows our construction’s upper bound is tight.
\end{itemize}

\begin{table*}[hbtp]
    \centering
    \caption{Resource costs for preparing the Dicke state $D(n,k)$ across $p$ QPUs, assuming for simplicity that $p$ divides $n$.}
    \begin{tabular}{c|cccc}
        Reference & Size & Depth & Qubits per QPU & Communication Complexity \\
        \hline
        \cite{ambainis2003quantum} ($p = 2$) & $O(2^n)$ & $O(2^n)$ & $n/2$ & $O(\log k)$ \\
        \cite{yuan2025depth} & $O(nk)$ & $O(k + \log(k) \log(n/k))$ & $n/p$ & $\Omega(k)$\\
        Ours & $O(nk)$ & $O(p^2 k + \log(k) \log(n/k))$ & $n/p + 2$ & $O(p \log k)$
    \end{tabular}
    \label{tab:comparison}
\end{table*}

\subsection{Paper Organization}

The remainder of the paper is organized as follows. In Section \ref{sec:pre}, we introduce key background concepts necessary for understanding the construction and analysis of our proposed distributed quantum circuit. In Section \ref{sec:2QPU}, to make our presentation more clear, we focus on $2$-QPU scenario present our circuit design for the  and establish its optimality by analyzing both its communication complexity and local circuit complexity. In Section \ref{sec:pQPU}, we generalize the aforementioned results to the general $p$-QPU case. In Section \ref{sec:lower}, we study the lower bound of communication complexity for dsitributed quantum state preparation. In Section \ref{sec:conclusion}, we conclude our paper and discuss some open and future problems.

\section{Preliminaries} \label{sec:pre}

In the remainder of the paper, we adopt the notation $[n] = \{0, 1, \dots ,n - 1\}$.

\subsection{Distributed Quantum Computing}

In this subsection, we briefly review fundamental concepts in quantum computing, and refer readers to \cite{nielsen2010quantum} for a comprehensive treatment. We then transition to DQC, where we formally define key metrics for evaluating the implementation cost of distributed quantum circuits.

In quantum computing, qubits serve as the fundamental units of information. The state of a single qubit can be a superposition of the computational basis states $\ket{0}$ and $\ket{1}$, and is written as $\alpha \ket{0} + \beta \ket{1}$, where $\alpha, \beta \in \mathbb{C}$ and $|\alpha|^2 + |\beta|^2 = 1$. More generally, the state of an $n$-qubit system can be expressed as
\begin{equation*}
    \sum_{x \in [2]^n} \alpha_x \ket{x}, \quad \alpha_x \in \mathbb{C},
\end{equation*}
where $x$ denotes a binary string of length $n$ and the amplitudes satisfy the normalization condition
\begin{equation*}
    \sum_{x \in [2]^n} |\alpha_x|^2 = 1.
\end{equation*}
Quantum gates form the elementary operations on qubits, with each gate represented by a unitary matrix that describes its action on the qubit space. The most widely adopted model for quantum computing is the circuit model, which consists of a sequence of quantum gates applied to selected qubits. Following standard quantum computing conventions, in this paper we restrict our focus to circuits constructed from a natural universal gate set, comprising CNOT gates and arbitrary single-qubit gates. This gate set is universal, meaning it can synthesize any arbitrary quantum operation with arbitrary precision. 

Within the framework of DQC, we partition a regular quantum circuit across QPUs. The parent circuit is decomposed into interconnected subcircuits that interact via inter-QPU communication gates---i.e., gates that act on qubits located on different QPUs. Similar to non-distributed quantum computing, two key metrics in DQC are circuit size and depth: the size captures the overall computational cost, while the depth reflects the temporal cost of execution, constrained by the coherence time of the qubits. In DQC, these metrics quantify local computation cost, as opposed to communication cost. Instead of considering the individual circuit of each QPU in isolation, we evaluate both metrics by treating the quantum circuits across different QPUs as a single, unified parent circuit. The formal definitions of circuit size and depth are as follows.

\begin{definition}[Circuit Size] \label{size}
    The \emph{size} of a distributed quantum circuit is defined as the total number of quantum gates in the unified parent circuit, excluding inter-QPU communication gates.
\end{definition}
\begin{definition}[Circuit Depth] \label{depth}
    The \emph{depth} of a distributed quantum circuit is the number of sequential gate layers in the unified parent circuit. Gates acting on disjoint qubit sets---capable of parallel execution---are grouped into the same layer.
\end{definition}

Quantum gates acting on qubits located across different QPUs incur substantially higher physical realization costs, owing to the inter-QPU communication required for their execution. Consequently, it is necessary to introduce a dedicated metric that quantifies the count of these non-local gates, rather than relying solely on conventional circuit size metrics. We formalize this metric as follows.

\begin{definition}[Communication Complexity] \label{cc}
    Let $C$ be a quantum circuit on $n$ qubits consisting only of one- and two-qubit gates. Let $f : [n] \to [p]$ be a partition that assigns the $n$ qubits to $p$ QPUs. The \emph{communication complexity} of $C$ with respect to the partition $f$, denoted by $\cc^f(C)$, is the number of two-qubit gates whose operands are assigned to different QPUs under $f$.

    For a quantum state $\ket{\psi}$, its \emph{communication complexity on $p$ QPUs} is defined as
    \begin{equation*}
        \cc_p(\ket{\psi}) := \min_{\substack{\mathrm{balanced} f:[n]\to [p] \\C: C\ket{0} = \ket{\psi}}} \cc^f(C),
    \end{equation*}
    where a partition $f$ is balanced if, for every $i\in[p]$, the size of its preimage $|f^{-1}(i)|$ is either $\lfloor n/p \rfloor$ or $\lceil n/p \rceil+1$.
\end{definition}

\subsection{Dicke State}

In this subsection, we present the formal definition of the Dicke state $D(n,k)$ and the concept of the Dicke unitary---a key operation that underpins the non-distributed Dicke states preparation.

\begin{definition}[Hamming Weight]
    The Hamming weight of a binary string $x$, denoted by $\ham(x)$, is the number of $1$s in $x$.
\end{definition}

\begin{definition}[Dicke State]
    Dicke state $D(n, k)$ is an $n$-qubit quantum state that is an equal superposition of all computational basis states with Hamming weight $k$.
    \begin{equation*}
        D(n, k) = \frac{1}{\sqrt{\binom{n}{k}}} \sum_{\substack{x \in \{0,1\}^n \\ \ham(x)=k}} \ket{x}.
    \end{equation*}
\end{definition}

Existing non-distributed approaches to preparing Dicke states do more than merely generate a single target state \cite{bartschi2019deterministic, bartschi2022short, yuan2025depth}. In fact, these methods implement a broader transformation that maps a set of basis states to the corresponding Dicke states. This motivates the following definition.

\begin{definition}[Dicke Unitary] \label{Dickeunitary}
    A Dicke unitary $U^n_k$ is defined by
    \begin{equation*}
        U^n_k \ket{0^{n - j} 1^j} = D(n, j), \quad \forall j \in [k + 1].
    \end{equation*}
\end{definition}

The state-of-the-art synthesis of the Dicke unitary in the non-distributed setting, which will be useful for our circuit design in Sections \ref{sec:2QPU} and \ref{sec:pQPU}, is due to \cite{yuan2025depth}. Specifically, they establish the following result. 

\begin{lemma}[\cite{yuan2025depth}] \label{non-disDicke25}
    For integers $n \ge k \ge 0$, the Dicke unitary $U^n_k$ can be implemented by a quantum circuit consisting of CNOT gates and single-qubit gates with circuit depth $O(k + \log k \log(n/k))$ and size $O(nk)$.
\end{lemma}

\subsection{Tensor and Canonical Polyadic Rank}

In this subsection, we introduce the definition of tensors and the canonical polyadic rank (CP-rank)---a metric that plays a pivotal role in our analysis of communication complexity lower bounds.
A tensor is a natural generalization of vectors and matrices, extending the structure of one-dimensional vectors and two-dimensional matrices to higher-order, multi-dimensional arrays.

\begin{definition}[Tensor] \label{deftensor}
    A $p$-order tensor $T = (T_{i_0 i_1 \dots i_{p-1}})$ is a $p$-dimensional array of numbers, where each index $i_j, j\in [p]$ corresponds to one mode of the tensor.
\end{definition}

When $p = 1$, a tensor reduces to a vector, and when $p = 2$, it coincides with a matrix. In this work, we need an extension of the classical notion of matrix rank to higher-order tensors.

\begin{definition}[Canonical Polyadic Rank (CP-rank)] \label{defcprank}
    For a $p$-order tensor $T$, a canonical polyadic decomposition (CP-decomposition) of $T$ is an expression of the form
    \begin{equation*}
        T = \sum_{i \in [r]} v_i^{(0)} \otimes v_i^{(1)} \otimes \cdots \otimes v_i^{(p-1)},
    \end{equation*}
    where $v_i^{(0)}, v_i^{(1)}, \dots, v_i^{(p-1)}$ are vectors for each $i \in [r]$, and $\otimes$ denotes the outer product.
    The canonical polyadic rank, or CP-rank, of $T$, denoted by $\CPrank(T)$, is the smallest integer $r$ for which such a decomposition exists.
\end{definition}

Intuitively, the CP-rank measures the minimum number of rank-one tensors needed to express $T$ as a sum of separable components, generalizing the concept of matrix rank to higher-order settings. In particular, when $p = 2$, the CP-rank coincides with the usual matrix rank: for any matrix $M$, we have $\rank(M) = \CPrank(M)$.

\subsection{Some Useful Circuit Implementations}

In this subsection, we collect several circuit implementation results that will be instrumental to our circuit design.

Lemma \ref{QSP} characterizes the asymptotically optimal cost for preparing an arbitrary quantum state, while Lemma \ref{ctrlQSP} extends this result to the controlled state preparation setting.

\begin{lemma}[State Preparation, \cite{sun2023asymptotically, yuan2023optimal}] \label{QSP}
    Given $m$ ancillary qubits, any $n$-qubit quantum state $\ket{\psi}$ can be prepared by a quantum circuit consisting of CNOT gates and single-qubit gates. The circuit has depth $O\left(n + \frac{2^{n}}{n+m}\right)$ and size $O(2^{n})$.
\end{lemma}

\begin{lemma}[Controlled State Preparation, \cite{yuan2023optimal}] \label{ctrlQSP}
    Let $k,m \ge 0$ and $n > 0$ be integers. Given a family of $n$-qubit quantum states $\{\ket{\psi_i}\}_{i \in [2^k]}$, the controlled state preparation
    \begin{equation*}
        \ket{i}\ket{0}^{\otimes n} \mapsto \ket{i}\ket{\psi_i}, \forall i \in [2^k],
    \end{equation*}
    can be implemented by a quantum circuit consisting of CNOT gates and single-qubit gates. The circuit has depth $O\left(n + k + \frac{2^{n+k}}{n+k+m}\right)$, size $O(2^{n+k})$, and uses $m$ ancillary qubits.
\end{lemma}

Lemma \ref{adder} characterizes the cost of implementing an addition circuit.

\begin{lemma}[Ripple-Carry Adder, \cite{remaud2024optimizing}] \label{adder}
    Let $\ket{a}$ and $\ket{b}$ denote two $n$-bit integers encoded in binary, and let $\ket{z}$ be a single-qubit register. The adder
    \begin{equation*}
        \ket{a}\ket{b}\ket{z} \mapsto \ket{a}\ket{a+b}\ket{z \oplus (a+b)_n},
    \end{equation*}
    where $(a+b)_n$ denotes the most significant (carry-out) bit of the sum $a+b$, can be implemented by a quantum circuit of depth $O(n)$ and size $O(n)$.
\end{lemma}

Lemma \ref{toffolisharingcontrol} characterizes the cost of a particular circuit structure composed of Toffoli gates, leveraging the notion of quantum fan-out gates as defined in Definition \ref{qfanout}. The cost of implementing quantum fan-out gates is established in Lemma \ref{fanoutCNOT}.

\begin{definition}[Quantum Fan-out Gate] \label{qfanout}
    For $n \ge 1$, the \emph{quantum fan-out gate} $F_n$ is the $(n+1)$-qubit unitary operator defined by
    \begin{equation*}
        F_n \ket{a}\ket{b_1}\cdots\ket{b_n} = \ket{a}\ket{a \oplus b_1}\cdots\ket{a \oplus b_n}.
    \end{equation*}
\end{definition}

\begin{lemma}[\cite{gokhale2021quantum}] \label{toffolisharingcontrol}
    Suppose there are $n$ Toffoli gates which share a common control qubit, and all other qubits involved are pairwise disjoint. Then the circuit can be implemented using fan-out gates, CNOT gates, and single-qubit gates with size $O(n)$ and depth $O(1)$.
\end{lemma}

\begin{lemma}[\cite{zi2025shallow}] \label{fanoutCNOT}
    The quantum fan-out gate $F_n$ can be implemented using $n$ CNOT gates with circuit depth $O(\log n)$ and without ancillary qubits.
\end{lemma}

\section{$2$-QPU Dicke State Preparation} \label{sec:2QPU}

In this section, we present a circuit for preparing the $n$-qubit Dicke state $D(n,k)$ across two QPUs, each containing $(n/2 + 2)$ qubits. Throughout this section, we assume for simplicity that $n$ is even.

\begin{theorem} \label{2QPUthm}
    For two QPUs, each equipped with $(n/2 + 2)$ qubits, there exists a quantum circuit that prepares the Dicke state $D(n,k)$ with circuit size $O(nk)$ and circuit depth $O(k + \log k \log(n/k))$. Moreover, the number of communication gates required between the two QPUs is $\lceil \log (k+1) \rceil$.
\end{theorem}

Remarkably, the circuit depth and size costs in Theorem \ref{2QPUthm} are directly inherited from Lemma \ref{non-disDicke25}. In other words, up to constant factors, our distributed preparation protocol introduces no additional overhead in circuit depth or size. The remainder of this section is devoted to a detailed description of the construction.

Intuitively, the communication complexity quantifies the amount of information that must be shared between the QPUs. For the Dicke state, without loss of generality, assume that $k \le n/2$. Our construction relies on the following decomposition:
\begin{equation*}
    \sqrt{\binom{n}{k}} D(n, k) = \sum_{j = 0}^k \sqrt{\binom{n/2}{j}} D(n/2, j) \sqrt{\binom{n/2}{k - j}} D(n/2, k - j).
\end{equation*}
Specifically, the QPUs coherently share a superposition over the index $j$. As a result, the distributed state preparation problem is reduced to the corresponding non-distributed preparation of local Dicke states.

To reduce the communication cost, the index $j$ should be encoded in binary using $\lceil\log (k + 1)\rceil$ qubits. Recall from Definition \ref{Dickeunitary} that the Dicke unitary expects the index $j$ in unary encoding. Therefore, the main challenge is to efficiently implement a transformation between these encodings. To this end, we introduce one-hot encoding as an intermediate representation. The formal definitions of the binary, one-hot, and unary encodings are given below.

\begin{definition} \label{encoding}
    Let $n \in \mathbb{N}$ and let $j \in [n + 1]$. There are three standard ways to encode the integer $j$ into a quantum state:
    \begin{itemize}
        \item \emph{Binary encoding}:
        \begin{equation*}
            \ket{j} = \ket{j_{\lceil \log_2 (n+1) \rceil - 1}\dots j_1 j_0},
        \end{equation*}
        where
        \begin{equation*}
            j = \sum_{l=0}^{\lceil \log_2 (n+1) \rceil - 1} j_l 2^l,
            \qquad j_l \in \{0,1\}.
        \end{equation*}

        \item \emph{One-hot encoding}:
        \begin{equation*}
            \ket{e^{n+1}_j} = \ket{0^{n-j}10^{j}},
        \end{equation*}
        which is an $(n+1)$-qubit state with a single excitation at position $j$.

        \item \emph{Unary encoding}:
        \begin{equation*}
            \ket{u^n_j} = \ket{0^{n-j}1^{j}},
        \end{equation*}
        consisting of $j$ consecutive ones followed by zeros.
    \end{itemize}
\end{definition}

Overall, our construction proceeds in three phases. In the first phase, $\QPU_0$ prepares a superposition of Hamming weights encoded in binary and transfers this state to $\QPU_1$. In the second phase, we convert the binary encoding into a one-hot encoding; this is the key step of our construction. Using only two ancillary qubits, we implement this transformation with polynomial size and depth. In the final phase, we transform the one-hot encoding into a unary encoding and invoke Lemma \ref{non-disDicke25} to complete the state preparation.

\textbf{Phase 1: Hamming-weight Distribution.}

\begin{figure*}[hbtp]
    \centering
    \begin{tikzpicture}
        \begin{yquant}
            qubit {} a[2];
            init {$\QPU_0$} (a);
            discard a;
            init {$\ket{0}$} a;
            qubit {} b[2];
            init {$\QPU_1$} (b);
            discard b;
            init {$\ket{0}$} b;
            
            ["north: $\lceil\log(k + 1)\rceil$" {font=\protect\footnotesize, inner sep=0pt}]
            slash a[0];
            ["north: others" {font=\protect\footnotesize, inner sep=0pt}]
            slash a[1];
            ["north: $\lceil\log(k + 1)\rceil$" {font=\protect\footnotesize, inner sep=0pt}]
            slash b[0];
            ["north: others" {font=\protect\footnotesize, inner sep=0pt}]
            slash b[1];
            align a;
            align b;
            box {$\ket{0}\mapsto \sum_j \alpha_j \ket{j}$} a[0];
            cnot b[0] | a[0];
            box {$\ket{j}\mapsto \ket{k-j}$} b[0];
            output {$\sum_j \alpha_j \ket{j}$} a[0];
            output {$\ket{0}$} a[1];
            output {$\sum_j \alpha_j \ket{k-j}$} b[0];
            output {$\ket{0}$} b[1];
        \end{yquant}
    \end{tikzpicture}
    \caption{\textbf{Phase 1}. Using the quantum state preparation circuit for a general state, $\QPU_0$ distributes the amplitudes according to $D(n/2, j)$ and $D(n/2, k - j)$, which are prepared on $\QPU_0$ and $\QPU_1$, respectively. Subsequently, $\QPU_0$ ``sends'' the $\ket{j}$ register to $\QPU_1$, thereby assigning the corresponding task that $\QPU_1$ must complete. After the minus circuit, the roles of $\QPU_0$ and $\QPU_1$ in the subsequent phases become symmetric; specifically, each performs the mapping $\ket{j} \mapsto D(n/2, j)$ using ancillary qubits.}
    \label{fig:phase1}
\end{figure*}
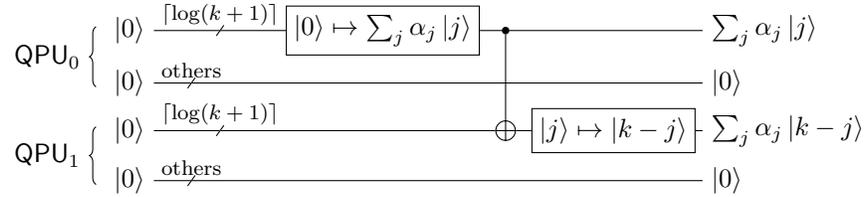

Figure \ref{fig:phase1} illustrates the three steps used to distribute the Hamming weights across the QPUs. Specifically, each computational basis state $\ket{j}$ is to be mapped to $D(n/2, j)$ in the subsequent phases. After the completion of this phase, no further communication between QPUs is required, and all operations within each QPU are performed symmetrically and in parallel. The following describes these steps in detail.
\begin{itemize}
    \item $\QPU_0$ prepares the state
    \begin{equation*}
        \sum_{j = 0}^k \alpha_j\ket{j} = 
        \frac{1}{\sqrt{\binom{n}{k}}} \sum_{j = 0}^k 
        \sqrt{\binom{n/2}{j}} \sqrt{\binom{n/2}{k - j}} \ket{j}
    \end{equation*}
    on $\lceil \log(k+1) \rceil$ qubits. The amplitudes correspond to the fraction of $(j, k-j)$ pairs distributed across the two QPUs. This step is implemented following Lemma \ref{QSP}. By utilizing all $(n/2 + 2)$ available qubits, the resulting circuit has size $O(2^{\log k}) = O(k)$ and depth $O(\log k + (2k/n)) = O(\log k)$.
    
    \item $\QPU_0$ copies the $\lceil \log(k+1) \rceil$ qubits to $\QPU_1$ using a single layer of $\lceil \log(k+1) \rceil$ CNOT gates. This step has circuit size $\lceil \log(k+1) \rceil$ and depth $1$.
    
    \item $\QPU_1$ performs the mapping $\ket{j} \mapsto \ket{k-j}$ on the $\lceil \log(k+1) \rceil$ qubits using the adder circuit described in Lemma \ref{adder}. This circuit has size and depth $O(\log k)$.
\end{itemize}

In summary, Phase 1 requires circuit size
\begin{equation} \label{phase1size}
    O(k) + \lceil \log(k+1) \rceil + O(\log k) = O(k).
\end{equation}
and depth
\begin{equation} \label{phase1depth}
    O(\log k) + 1 + O(\log k) = O(\log k).
\end{equation}
The communication complexity incurred in this phase is $\lceil \log(k+1) \rceil$.

\textbf{Phase 2: Binary Encoding to One-hot Encoding.}

From this phase, we focus on a specific QPU, noting that the others operate similarly. In general, we recursively convert the binary encoding into the one-hot encoding, proceeding from the lower bits to the higher bits.

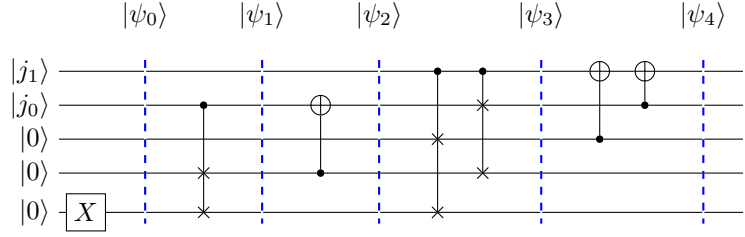
\begin{figure*}
    \centering
    \begin{tikzpicture}[label distance=4mm]
    \begin{yquant}[operators/every barrier/.append style={blue, thick}]
        qubit {} q[5];
        init {$\ket{j_1}$} q[0];
        init {$\ket{j_0}$} q[1];
        init {$\ket{0}$} q[2, 3, 4];

        x q[4];
        ["$\ket{\psi_0}$"]
        barrier (-);
        swap (q[3, 4]) | q[1];
        ["$\ket{\psi_1}$"]
        barrier (-);
        cnot q[1] | q[3];
        ["$\ket{\psi_2}$"]
        barrier (-);
        swap (q[2, 4]) | q[0];
        swap (q[1, 3]) | q[0];
        ["$\ket{\psi_3}$"]
        barrier (-);
        cnot q[0] | q[2];
        cnot q[0] | q[1];
        ["$\ket{\psi_4}$"]
        barrier (-);
    \end{yquant}
\end{tikzpicture}
    \caption{An example of \textbf{Phase~2} for $k = 3$. The key observation is that, under the one-hot encoding, the operation $+2^l$, and hence $+ j_l \times 2^l$, can be implemented straightforwardly using swap and controlled-swap gates, respectively. Moreover, the state of the control qubit $\ket{j_l}$ can be inferred by detecting the presence of a swap operation, which reduces the number of ancillary qubits required from $O(\log k)$ to $O(1)$.}
    \label{fig:phase2example}
\end{figure*}

Take $k = 3$ as an illustrative example. The corresponding circuit is shown in Figure~\ref{fig:phase2example}, where qubits unaffected by the operations are omitted for clarity. Let
\begin{equation*}
    j = j_1 \times 2^1 + j_0 \times 2^0, \qquad j_0, j_1 \in [2].
\end{equation*}
Our goal is to implement the transformation
\begin{equation*}
    \ket{j_1 j_0 0 0 0} \mapsto \ket{0 e^4_j}.
\end{equation*}

Recalling Definition~\ref{encoding}, we begin with
\begin{equation*}
    \ket{\psi_0} = \ket{j_1 j_0 0 0 1} = \ket{j_1 j_0 0 0 e^1_0}.
\end{equation*}
The controlled-swap gate effectively performs an addition $+ j_0 \times 2^0$ in the one-hot encoding, resulting in
\begin{equation*}
    \ket{\psi_1} = \ket{j_1 j_0 0 e^2_{j_0}}.
\end{equation*}
The occurrence of the swap indicates that $j_0 = 1$. Consequently, a CNOT gate can be applied to reset the qubit $\ket{j_0}$ to $\ket{0}$, yielding
\begin{equation*}
    \ket{\psi_2} = \ket{j_1 0 0 e^2_{j_0}}.
\end{equation*}
Thus, the qubit originally storing $\ket{j_0}$ is freed and can be reused as part of the one-hot encoding register.

Next, we apply two controlled-swap gates to realize the addition $+\, j_1 \times 2^1$, and then eliminate $\ket{j_1}$ in an analogous manner. For simplicity, we summarize the results of these steps as
\begin{align*}
    \ket{\psi_3} &= \ket{j_1 e^4_{j_1 j_0}}, \\
    \ket{\psi_4} &= \ket{0 e^4_{j_1 j_0}}.
\end{align*}

\begin{figure}[hbtp]
    \centering
    \begin{tikzpicture}
        \begin{yquant}
            qubit {$\ket{j_{\lceil\log(k + 1)\rceil - 1}}$} c;
            qubit {$\vdots$} c[+1]; discard c[1];
            qubit {$\ket{j_l}$} c[+1];
            
            qubit {$\ket{0}$} t;
            qubit {$\ket{0}$} t[+1];
            qubit {$\ket{e^{2^{l}}_{j_{l-1}\dots j_0}}$} t[+1];
    
            ["north: others" {font=\protect\footnotesize, inner sep=0pt}]
            slash t[0];
            ["north: $2^{l}$" {font=\protect\footnotesize, inner sep=0pt}]
            slash t[1];
            ["north: $2^{l}$" {font=\protect\footnotesize, inner sep=0pt}]
            slash t[2];
            ["south: $x\leftrightarrow x + 2^l$" {font=\protect\footnotesize, inner sep=0pt}]
            swap (t[1, 2]) | c[2];
            cnot c[2] | t[1];
            
            output {$\ket{j_{\lceil\log(k + 1)\rceil - 1}}$} c[0];
            output {$\ket{0}$} c[2];
            output {$\ket{0}$} t[0];
            output {$\ket{e^{2^{l + 1}}_{j_l\dots j_0}}$} (t[1, 2]);
        \end{yquant}
    \end{tikzpicture}
    \caption{The subcircuit handling $\ket{j_l}$ in \textbf{Phase 2}. After each such subcircuit, the number of qubits used for the one-hot encoding doubles, while the number of qubits used for the binary encoding decreases by one. The qubit freed from the binary encoding is repurposed for the one-hot encoding, except for the most significant qubit, $\ket{j_{\lceil \log (k+1)\rceil - 1}}$. In addition to the qubit representing $0$ in the one-hot encoding, a total of two ancillary qubits are required.}
    \label{fig:phase2}
\end{figure}
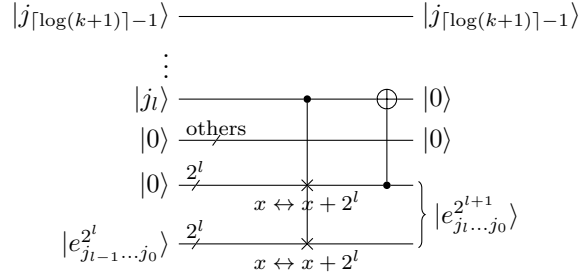

Returning to the general case, Figure \ref{fig:phase2} illustrates the circuit segment that processes the qubit $\ket{j_l}$ for some $l \in [\lceil \log (k+1) \rceil]$. The $2^l$ controlled-swap gates collectively implement the operation $+ j_l \cdot 2^l$ on the register $\ket{e^{2^l}_{j_{l-1}\dots j_1 j_0}}$. Subsequently, $2^l$ CNOT gates are applied to eliminate the qubit $\ket{j_l}$, as analyzed in the preceding example. Iterating this procedure from $\ket{j_0}$ through $\ket{j_{\lceil \log (k+1)\rceil - 1}}$ completes this phase.

Each controlled-swap gate is decomposed into three Toffoli gates, and Lemmas \ref{toffolisharingcontrol} and~\ref{fanoutCNOT} are used to optimize their implementation. As a result, the controlled-swap layer can be realized with circuit size $O(2^l)$ and depth $O(l)$. The subsequent CNOT layer corresponds to a fan-in operation; hence, Lemma \ref{fanoutCNOT} applies directly, yielding circuit size $O(2^l)$ and depth $O(l)$. Summing over all $l \in [\lceil \log (k+1) \rceil]$, we obtain a total circuit size
\begin{equation} \label{phase2size}
    \sum_{l = 0}^{\lceil \log (k+1) \rceil} O(2^l) = O(k)
\end{equation}
and depth 
\begin{equation} \label{phase2depth}
    \sum_{l = 0}^{\lceil \log (k+1) \rceil} O(l) = O(\log^2 k)
\end{equation}
for this phase.

\textbf{Phase 3: Local Non-distributed Preparation.}

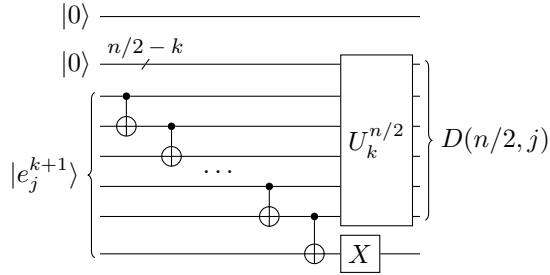
\begin{figure}[hbtp]
    \centering
    \begin{tikzpicture}
        \begin{yquant}
            qubit {$\ket{0}$} a;
            qubit {$\ket{0}$} r;
            [name=ypos]
            qubit {} q[6];
            
            init {$\ket{e^{k + 1}_j}$} (q);
            ["north: $n/2-k$" {font=\protect\footnotesize, inner sep=0pt}]
            slash r;
            cnot q[1] | q[0];
            [name=left]
            cnot q[2] | q[1];
            hspace {7mm} -;
            [name=right]
            cnot q[4] | q[3];
            cnot q[5] | q[4];
            x q[5];
            
            box {$U^{n/2}_k$} (r, q[0-4]);
            output {$D(n/2, j)$} (r, q[0-4]);
        \end{yquant}
        \path (left |- ypos-2) -- (right |- ypos-3) node[midway] {$\dots$};
    \end{tikzpicture}
    \caption{\textbf{Phase 3}. The downstairs CNOT gates transform the one-hot encoding into the unary encoding, introducing an additional $1$ due to the representation of $0$ in the one-hot encoding. The subsequent $X$ gate removes this extra $1$. After this correction, the Dicke unitary, implemented in a non-distributed manner, can be applied directly.}
    \label{fig:phase3}
\end{figure}

In the final phase, illustrated in Figure \ref{fig:phase3}, we first convert the one-hot encoding $\ket{e^{k+1}_j}$ into the unary encoding $\ket{u^k_j}$. This conversion is implemented using $k$ CNOT gates. Since Definition \ref{encoding} represents the one-hot encoding on $(k+1)$ qubits and the unary encoding on $k$ qubits, the extra qubit is removed by a final $X$ gate. For clarity, we explicitly describe the mapping as
\begin{equation*}
    \ket{e^{k+1}_j} \mapsto \ket{0^{k-j}1^{j+1}} \mapsto \ket{u^k_j 0}.
\end{equation*}
This transformation has circuit size $O(k)$ and depth $O(k)$.

Subsequently, we apply Lemma \ref{non-disDicke25} and use the Dicke unitary $U^n_k$ to prepare the state $D(n/2, j)$. The resulting circuit depth and size are directly inherited from the non-distributed Dicke state preparation.

The cost of this phase consists of the circuit size
\begin{equation} \label{phase3size}
    O(k) + O(nk) = O(nk).
\end{equation}
and depth
\begin{equation} \label{phase3depth}
    O(k) + O(k + \log k \log(n/k)) = O(k + \log k \log(n/k))
\end{equation}

Summing the contributions from equations \ref{phase1size}, \ref{phase2size}, and \ref{phase3size}, we obtain the total size of the circuit:
\begin{equation*}
    O(k) + O(k) + O(nk) = O(nk).
\end{equation*}
Similarly, summing the bounds from equations \ref{phase1depth}, \ref{phase2depth}, and \ref{phase3depth}, the total depth of the circuit is
\begin{equation*}
    O(\log k) + O(\log^2 k) + O(k + \log k \log(n/k)) = O(k + \log k \log(n/k)),
\end{equation*}
Communication occurs only in Phase 1, with complexity $\lceil \log (k+1) \rceil$.

In summary, the circuit depth and size incurred by distributing the Hamming weight and performing the encoding conversion do not exceed those of the non-distributed Dicke unitary, as established in Lemma \ref{non-disDicke25}. This completes the proof of Theorem \ref{2QPUthm}.

\section{General $p$-QPU Dicke State Preparation} \label{sec:pQPU}

In this section, we extend our circuit construction to the general case of $p$ QPUs. For simplicity, we assume throughout the remainder of the paper that $p$ divides $n$ evenly, i.e., $n=kp$ for some positive integer $k$. The results proved for the divisible case can be straightforwardly generalized to the non-divisible case by substituting $n/p$ with $\lceil n / p\rceil$ in Table \ref{tab:comparison}.

\begin{theorem} \label{pQPUthm}
    For $p$ QPUs, each equipped with $(n/p+2)$ qubits, there exists a quantum circuit that prepares the Dicke state $D(n,k)$ with circuit size $O(nk)$ and circuit depth $O(p^2 k + \log k \log(n/k))$. Moreover, the number of communication gates required between the QPUs is $O(p \log k)$.
\end{theorem}

\begin{proof}

The proof is structured into three parts. In Part 1, we compare the cases $p = 2$ and $p\ge 3$, noting that the $p\ge 3$ case differs from the $p = 2$ case only in Phase 1 of the $p = 2$ construction. By emphasizing these differences, we outline the core idea of our distributed circuit construction. In Part 2, we detail how to adapt Phase 1---the Hamming-weight-distributed phase---from the $p = 2$ case to the general $p\ge 3$ scenario. In Part 3, we analyze the total circuit cost (size, depth, and communication complexity) and thereby complete the proof.

\textbf{Part 1.}

Compared with the construction in Section \ref{sec:2QPU}, the overall construction differs only in Phase 1. Specifically, we must now coordinate the actions of the $p$ QPUs to prepare a coherent superposition of the form
$\ket{j_0}\ket{j_1}\cdots\ket{j_{p-1}}$
distributed across the $p$ QPUs. Once this distribution is achieved, the remaining phases proceed in close analogy to the 2-QPU case and admit a similar cost analysis.

Unlike the 2-QPU case, when $p\ge 3$ we can no longer assume that $k \le n/p$. We therefore define
\begin{equation*}
    k' = \min(n/p, k),
\end{equation*}
which represents the maximum Hamming weight which can be assigned to a single QPU. Let $J = (j_i)_{i\in[p]}$. The feasible region for the distribution of Hamming weight is then given by
\begin{equation*}
    F_p = \left\{ J : j_i \in [k' + 1]\ \forall i \in [p], \ \text{and } \sum_{i\in[p]} j_i = k \right\}.
\end{equation*}

Now we decompose the Dicke state $D(n,k)$ into $p$ components as
\begin{equation*}
    \sqrt{\binom{n}{k}} D(n,k) = \sum_{J\in F_p}
    \sqrt{\binom{n/p}{j_0}} D(n/p, j_0)
    \cdots
    \sqrt{\binom{n/p}{j_{p - 1}}} D(n/p, j_{p-1}).
\end{equation*}

\textbf{Part 2.}

In what follows, we detail the modified Hamming-weight distribution phase. We propagate communication sequentially across the QPUs---a design choice well suited to a linear QPU connection topology. A key observation underpins this approach: each QPU only requires knowledge of the total Hamming weight prepared by the preceding QPUs, rather than their individual contributions. This intentional design significantly reduces the communication complexity of the distributed circuit.

\begin{figure*}[hbtp]
\centering
\begin{tikzpicture}
    \begin{yquant}[register/separation=3mm, every nobit output/.style={}]
        qubit {} q1[2];
        init {$\QPU_{0}$} (q1);
        discard q1;
        init {$\ket{0}$} q1;
        qubit {} q2[3];
        init {$\QPU_{1}$} (q2);
        discard q2;
        init {$\ket{0}$} q2;
        qubit {$\vdots$} e1; discard e1;
        qubit {} qi[3];
        init {$\QPU_{i}$} (qi);
        discard qi;
        init {$\ket{0}$} qi;
        qubit {} qi1[3];
        init {$\QPU_{i + 1}$} (qi1);
        discard qi1;
        init {$\ket{0}$} qi1;
        qubit {$\vdots$} e2; discard e2;
        qubit {} p1[3];
        init {$\QPU_{p - 1}$} (p1);
        discard p1;
        init {$\ket{0}$} p1;
        
        ["north: $\lceil\log(k + 1)\rceil$" {font=\protect\footnotesize, inner sep=0pt}]
        slash q1[0];
        ["north: others" {font=\protect\footnotesize, inner sep=0pt}]
        slash q1[1];
        ["north: $\lceil\log(k + 1)\rceil$" {font=\protect\footnotesize, inner sep=0pt}]
        slash q2[0];
        ["north: $\lceil\log(k + 1)\rceil$" {font=\protect\footnotesize, inner sep=0pt}]
        slash q2[1];
        ["north: others" {font=\protect\footnotesize, inner sep=0pt}]
        slash q2[2];
        ["north: $\lceil\log(k + 1)\rceil$" {font=\protect\footnotesize, inner sep=0pt}]
        slash qi[0];
        ["north: $\lceil\log(k + 1)\rceil$" {font=\protect\footnotesize, inner sep=0pt}]
        slash qi[1];
        ["north: others" {font=\protect\footnotesize, inner sep=0pt}]
        slash qi[2];
        ["north: $\lceil\log(k + 1)\rceil$" {font=\protect\footnotesize, inner sep=0pt}]
        slash qi1[0];
        ["north: $\lceil\log(k + 1)\rceil$" {font=\protect\footnotesize, inner sep=0pt}]
        slash qi1[1];
        ["north: others" {font=\protect\footnotesize, inner sep=0pt}]
        slash qi1[2];
        ["north: $\lceil\log(k + 1)\rceil$" {font=\protect\footnotesize, inner sep=0pt}]
        slash p1[0];
        ["north: $\lceil\log(k + 1)\rceil$" {font=\protect\footnotesize, inner sep=0pt}]
        slash p1[1];
        ["north: others" {font=\protect\footnotesize, inner sep=0pt}]
        slash p1[2];
        align q1, q2, e1, qi, qi1, e2, p1;

        box {QSP} q1[0];
        cnot q2[1] | q1[0];

        align q1, q2, e1, qi, qi1, e2, p1;
        hspace {0mm} e1;
        [name=v1]
        box {$V_1$} (q2[0], q2[1]);
        [name=edots1]
        text {$\vdots$} e1;
        \draw[decorate, decoration={snake, amplitude=.25mm, segment length=5pt}] (v1) -- (edots1);

        align q1, q2, e1, qi, qi1, e2, p1;
        hspace {1.5mm} e1;
        [name=edots2]
        text {$\vdots$} e1;
        [name=vi]
        box {$V_{i - 1}$} qi[1];
        \draw[decorate, decoration={snake, amplitude=.25mm, segment length=5pt}] (edots2) -- (vi);

        box {$V_{i}$} (qi[0], qi[1], qi1[1]);

        align q1, q2, e1, qi, qi1, e2, p1;
        [name=vi1]
        box {$V_{i + 1}$} (qi1[0], qi1[1]);
        hspace {1.5mm} e2;
        [name=edotsi]
        text {$\vdots$} e2;
        \draw[decorate, decoration={snake, amplitude=.25mm, segment length=5pt}] (edotsi) -- (vi1);

        align q1, q2, e1, qi, qi1, e2, p1;
        hspace {1.5mm} e2;
        [name=edotsp1]
        text {$\vdots$} e2;
        [name=vp1]
        box {$V_{p - 2}$} p1[1];
        \draw[decorate, decoration={snake, amplitude=.25mm, segment length=5pt}] (edotsp1) -- (vp1);

        box {$A$} (p1[0], p1[1]);
        
        align q1, q2, e1, qi, qi1, e2, p1;
        hspace {2.5mm} e2;
        [name=edotsp1]
        text {$\vdots$} e2;
        [name=vp1]
        box {$W_{p - 2}$} p1[1];
        \draw[decorate, decoration={snake, amplitude=.25mm, segment length=5pt}] (edotsp1) -- (vp1);

        align q1, q2, e1, qi, qi1, e2, p1;
        [name=vi1]
        box {$W_{i + 1}$} (qi1[0], qi1[1]);
        hspace {2.5mm} e2;
        [name=edotsi]
        text {$\vdots$} e2;
        \draw[decorate, decoration={snake, amplitude=.25mm, segment length=5pt}] (edotsi) -- (vi1);

        box {$W_{i}$} (qi[0], qi[1], qi1[1]);

        align q1, q2, e1, qi, qi1, e2, p1;
        hspace {2.5mm} e1;
        [name=edots2]
        text {$\vdots$} e1;
        [name=vi]
        box {$W_{i - 1}$} qi[1];
        \draw[decorate, decoration={snake, amplitude=.25mm, segment length=5pt}] (edots2) -- (vi);

        align q1, q2, e1, qi, qi1, e2, p1;
        [name=v1]
        box {$W_1$} (q2[0], q2[1]);
        hspace {1.25mm} e1;
        [name=edots1]
        text {$\vdots$} e1;
        \draw[decorate, decoration={snake, amplitude=.25mm, segment length=5pt}] (v1) -- (edots1);

        cnot q2[1] | q1[0];
    \end{yquant}
\end{tikzpicture}
    \caption{Modified \textbf{Phase 1} for the general $p$-QPU case.
    The first $\lceil \log (k+1) \rceil$ qubits in $\QPU_i$ are used to encode $\ket{j_i}$, which specifies the local state $D(n/p, j_i)$ to be prepared in the subsequent phases. $\QPU_0$ allocates the amplitudes of $\ket{j_0}$ using the subcircuit QSP. For $i = 1, 2, \dots, p-2$, $\QPU_i$ allocates the amplitudes of $\ket{j_i}$ and computes the partial sum $\ket{J_i}$ within the subcircuit $V_i$ (See Figure \ref{fig:pQPUphase1V}.), using $\ket{J_{i-1}}$ as input, where the register $\ket{J_{i-1}}$ is copied to a second set of $\lceil \log (k+1) \rceil$ qubits by $\QPU_{i-1}$. The final processor, $\QPU_{p-1}$, only needs to compute $\ket{j_{p-1}}$ via the subcircuit $A$, which implements the operation $k - J_{p-2}$. Finally, the temporarily stored partial sums are uncomputed using the subcircuits $W_i$ (See Figure \ref{fig:pQPUphase1W}.) for $i = p-2, p-3, \dots, 1$, followed by a final layer of $\lceil \log (k+1) \rceil$ CNOT gates to complete the cleanup.}
    \label{fig:pQPUphase1}
\end{figure*}
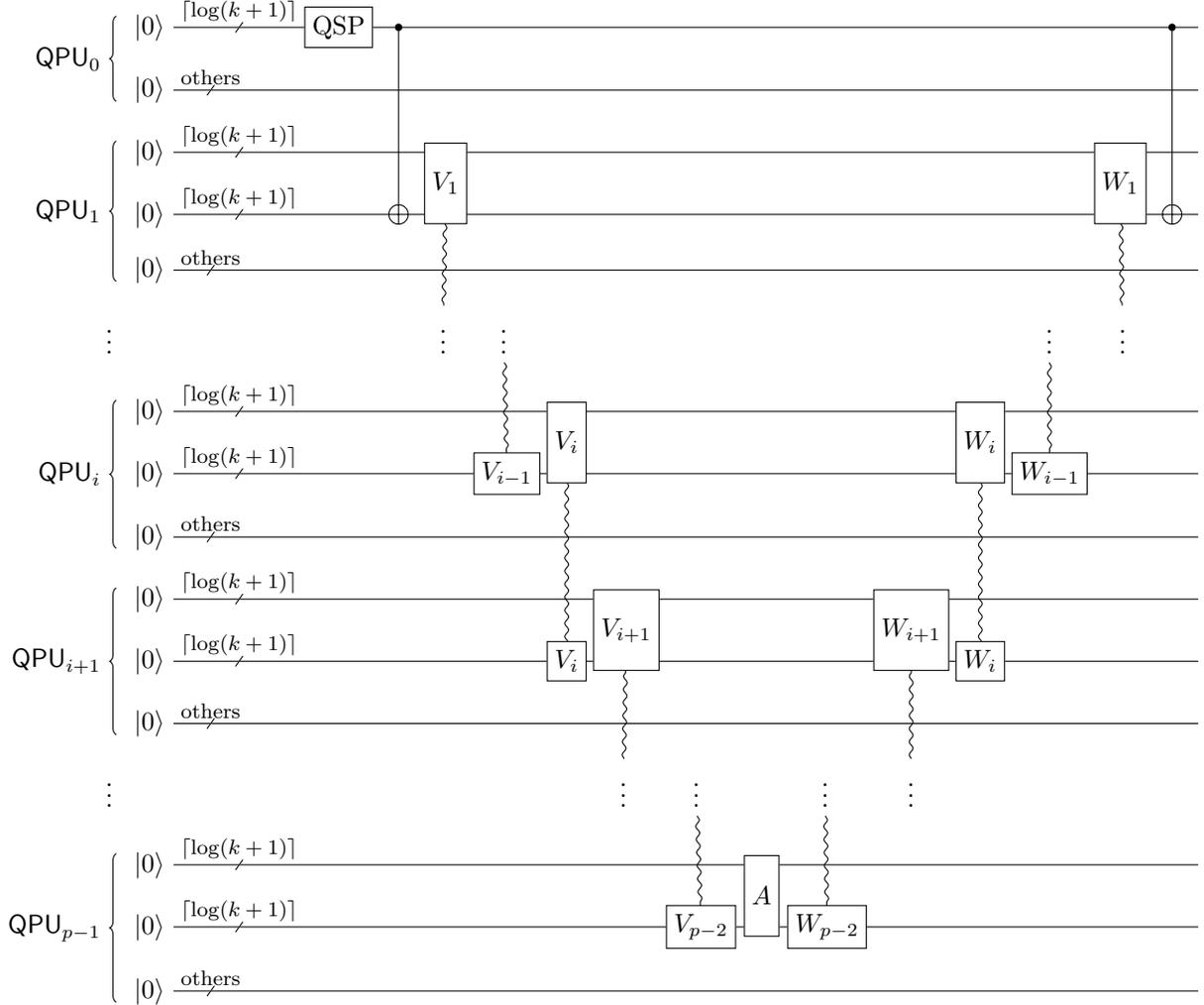

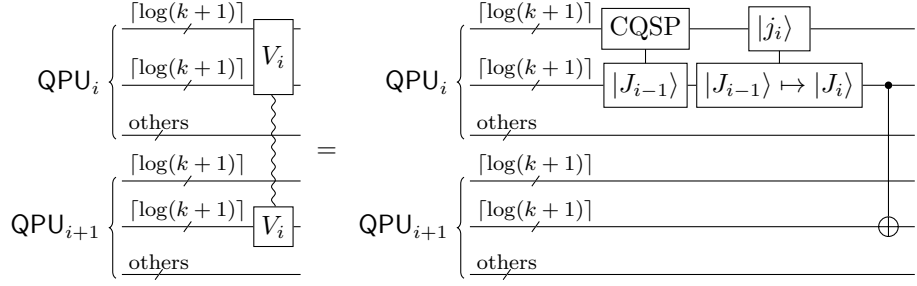
\begin{figure*}[hbtp]
    \centering
    \begin{tikzpicture}
    \begin{yquantgroup}
        \registers{
            qubit {} q1[3];
            qubit {} q2[3];
        }
        \circuit{
            init {$\QPU_{i}$} (q1);
            init {$\QPU_{i + 1}$} (q2);

            ["north: $\lceil\log(k + 1)\rceil$" {font=\protect\footnotesize, inner sep=0pt}]
            slash q1[0];
            ["north: $\lceil\log(k + 1)\rceil$" {font=\protect\footnotesize, inner sep=0pt}]
            slash q1[1];
            ["north: others" {font=\protect\footnotesize, inner sep=0pt}]
            slash q1[2];
            ["north: $\lceil\log(k + 1)\rceil$" {font=\protect\footnotesize, inner sep=0pt}]
            slash q2[0];
            ["north: $\lceil\log(k + 1)\rceil$" {font=\protect\footnotesize, inner sep=0pt}]
            slash q2[1];
            ["north: others" {font=\protect\footnotesize, inner sep=0pt}]
            slash q2[2];
            
            box {$V_{i}$} (q1[0], q1[1], q2[1]);
        }
        \equals
        \circuit{
            init {$\QPU_{i}$} (q1);
            init {$\QPU_{i + 1}$} (q2);

            ["north: $\lceil\log(k + 1)\rceil$" {font=\protect\footnotesize, inner sep=0pt}]
            slash q1[0];
            ["north: $\lceil\log(k + 1)\rceil$" {font=\protect\footnotesize, inner sep=0pt}]
            slash q1[1];
            ["north: others" {font=\protect\footnotesize, inner sep=0pt}]
            slash q1[2];
            ["north: $\lceil\log(k + 1)\rceil$" {font=\protect\footnotesize, inner sep=0pt}]
            slash q2[0];
            ["north: $\lceil\log(k + 1)\rceil$" {font=\protect\footnotesize, inner sep=0pt}]
            slash q2[1];
            ["north: others" {font=\protect\footnotesize, inner sep=0pt}]
            slash q2[2];

            [name=cqsp1]
            box {\Ifnum\idx<1 CQSP\Else $\ket{J_{i - 1}}$\Fi} q1[0], q1[1];
            \draw (cqsp1-0) -- (cqsp1-1);
            [name=add]
            box {\Ifnum\idx<1 $\ket{j_i}$ \Else $\ket{J_{i - 1}}\mapsto \ket{J_i}$\Fi} q1[0], q1[1];
            \draw (add-0) -- (add-1);
            cnot q2[1] | q1[1];
        }
    \end{yquantgroup}
    \end{tikzpicture}
    \caption{The subcircuit $V_i, i = 1, 2, \dots, p - 2$ in Figure \ref{fig:pQPUphase1}. Conditioned on the total Hamming weight $\ket{J_{i-1}}$ accumulated from the previous QPUs, $\QPU_i$ allocates the amplitudes of $\ket{j_i}$ according to their frequency in $F_p$. The register $\ket{j_i}$ is then combined with $\ket{J_{i-1}}$ to form $\ket{J_i}$, which is subsequently forwarded to $\QPU_{i+1}$.}
    \label{fig:pQPUphase1V}
\end{figure*}

\begin{figure*}[hbtp]
    \centering
    \begin{tikzpicture}
    \begin{yquantgroup}
        \registers{
            qubit {} q1[3];
            qubit {} q2[3];
        }
        \circuit{
            init {$\QPU_{i}$} (q1);
            init {$\QPU_{i + 1}$} (q2);

            ["north: $\lceil\log(k + 1)\rceil$" {font=\protect\footnotesize, inner sep=0pt}]
            slash q1[0];
            ["north: $\lceil\log(k + 1)\rceil$" {font=\protect\footnotesize, inner sep=0pt}]
            slash q1[1];
            ["north: others" {font=\protect\footnotesize, inner sep=0pt}]
            slash q1[2];
            ["north: $\lceil\log(k + 1)\rceil$" {font=\protect\footnotesize, inner sep=0pt}]
            slash q2[0];
            ["north: $\lceil\log(k + 1)\rceil$" {font=\protect\footnotesize, inner sep=0pt}]
            slash q2[1];
            ["north: others" {font=\protect\footnotesize, inner sep=0pt}]
            slash q2[2];
            
            box {$W_{i}$} (q1[0], q1[1], q2[1]);
        }
        \equals
        \circuit{
            init {$\QPU_{i}$} (q1);
            init {$\QPU_{i + 1}$} (q2);

            ["north: $\lceil\log(k + 1)\rceil$" {font=\protect\footnotesize, inner sep=0pt}]
            slash q1[0];
            ["north: $\lceil\log(k + 1)\rceil$" {font=\protect\footnotesize, inner sep=0pt}]
            slash q1[1];
            ["north: others" {font=\protect\footnotesize, inner sep=0pt}]
            slash q1[2];
            ["north: $\lceil\log(k + 1)\rceil$" {font=\protect\footnotesize, inner sep=0pt}]
            slash q2[0];
            ["north: $\lceil\log(k + 1)\rceil$" {font=\protect\footnotesize, inner sep=0pt}]
            slash q2[1];
            ["north: others" {font=\protect\footnotesize, inner sep=0pt}]
            slash q2[2];

            cnot q2[1] | q1[1];
            [name=sub]
            box {\Ifnum\idx<1 $\ket{j_i}$ \Else $\ket{J_i}\mapsto \ket{J_{i - 1}}$\Fi} q1[0], q1[1];
            \draw (sub-0) -- (sub-1);
        }
    \end{yquantgroup}
    \end{tikzpicture}
    \caption{The subcircuit $W_i, i = 1, 2, \dots, p - 2$ in Figure \ref{fig:pQPUphase1}. After eliminating $\ket{J_i}$ in $\QPU_{i+1}$, the partial sum $\ket{J_i}$ stored in $\QPU_i$ is restored to $\ket{J_{i-1}}$ via a subtraction subcircuit, enabling further uncomputation.}
    \label{fig:pQPUphase1W}
\end{figure*}
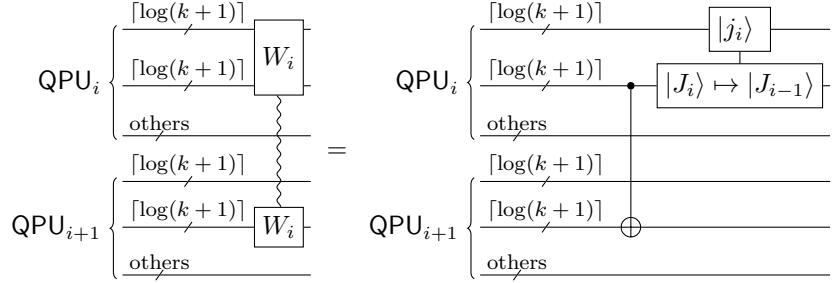

For notation convenience, define the partial sums $J_i = j_0 + j_1 + \cdots j_i$ for $i \in [p]$, and set
\begin{equation*}
    \alpha_{J_i}
    = \sum_{\substack{j_{i+1}, \dots, j_{p-1} \\ \text{such that } J\in F_p}}
    \binom{n/p}{j_{i+1}}
    \cdots
    \binom{n/p}{j_{p-1}},
\end{equation*}
which is well-defined by the definition of $F_p$. Moreover, we will henceforth use $k$ in place of $k'$. Since $k' = \min(n/p, k)\le k$, this substitution does not affect our upper-bound analysis.

Figure \ref{fig:pQPUphase1} depicts the circuit implementing the modified Phase 1 for the general $p$-QPU case. The procedure is carried out according to the following steps.

\begin{itemize}
    \item $\QPU_0$ prepares the state
    \begin{equation*}
    \frac{1}{\sqrt{\binom{n}{k}}} \sum_{\substack{j_0\in J \\ \text{such that } \exists J\in F_p}}
    \sqrt{\binom{n/p}{j_0}} \sqrt{\alpha_{J_0}} \ket{j_0}
    \end{equation*}
    using $\lceil \log(k+1) \rceil$ qubits, which is depicted in Figure \ref{fig:pQPUphase1} as the gate $QSP$. It then copies these qubits to $\QPU_1$ via a single layer of $\lceil \log(k+1) \rceil$ CNOT gates. By utilizing all available $(n/p+2)$ qubits, Lemma~\ref{QSP} implies that the resulting circuit has size $O(k)$ and depth $O(\log k)$.
    
    \item For $i = 1, 2, \dots, p-2$, $\QPU_i$ applies Lemma~\ref{ctrlQSP} to implement the controlled quantum state preparation
    \begin{equation*}
        \ket{J_{i-1}} \ket{0} 
        \mapsto
        \ket{J_{i-1}} \frac{1}{\sqrt{\alpha_{J_{i-1}}}}
        \sum_{\substack{j_i\in J \\ \text{such that } \exists J\in F_p}}
        \sqrt{\binom{n/p}{j_i}} \sqrt{\alpha_{J_i}} \ket{j_i}
    \end{equation*}
    on $2\lceil \log(k+1) \rceil$ qubits. The register $\ket{j_i}$ is then added to $\ket{J_{i-1}}$ to obtain $\ket{J_i}$, namely,
    \begin{equation*}
        \ket{J_{i-1}} \ket{j_i}
        \mapsto
        \ket{J_{i-1} + j_i} \ket{j_i}
        = \ket{J_i} \ket{j_i}.
    \end{equation*}
    Subsequently, $\ket{J_i}$ is copied to $\QPU_{i+1}$ using a single layer of $\lceil \log(k+1) \rceil$ CNOT gates. This operation is depicted in Figure \ref{fig:pQPUphase1} as the gate $V_i$, with further details shown in Figure \ref{fig:pQPUphase1V}. Again, by exploiting all $(n/p+2)$ available qubits, Lemmas \ref{ctrlQSP} and \ref{adder} together imply that the resulting circuit has size $O(k^2)$ and depth $O(pk)$.
    
    \item Upon reaching $\QPU_{p-1}$, the copied register $\ket{J_{p-2}}$ suffices to determine the remaining Hamming weight. Therefore, $\QPU_{p-1}$ only needs to implement
    \begin{equation*}
        \ket{J_{p-2}} \ket{0}
        \mapsto
        \ket{J_{p-2}} \ket{k - J_{p-2}}
        =\ket{J_{p-2}} \ket{j_{p-1}},
    \end{equation*}
    thereby completing the preparation, which is depicted in Figure \ref{fig:pQPUphase1} as the gate $A$. By Lemma~\ref{adder}, both the circuit size and depth are $O(\log k)$.
    
    \item Finally, for $i = p-2, p-3, \dots, 1$, we uncompute the auxiliary registers $\ket{J_i}$ stored in $\QPU_{i + 1}$. Specifically, we first erase the register $\ket{J_i}$ in $\QPU_{i + 1}$ using its copy in $\QPU_i$ via a layer of $\lceil \log(k+1) \rceil$ CNOT gates. Then we apply
    \begin{equation*}
        \ket{J_i} \ket{j_i}
        \mapsto
        \ket{J_i - j_i} \ket{j_i}
        = \ket{J_{i-1}} \ket{j_i}
    \end{equation*}
    within $\QPU_i$ to prepare for the next iteration. This operation is depicted in Figure \ref{fig:pQPUphase1} as the gate $W_i$, with further details shown in Figure \ref{fig:pQPUphase1W}.  Upon reaching $\QPU_0$, a final layer of $\lceil \log(k+1) \rceil$ CNOT gates suffices to complete the uncomputation. By Lemma \ref{adder}, this entire uncomputation step incurs circuit size and depth $O(\log k)$.
\end{itemize}

In this phase, our circuit achieves the size of
\begin{equation} \label{mphase1size}
    O(k) + O(pk^2) + O(\log k) + O(p\log k) = O(pk^2).
\end{equation}
and the depth of
\begin{equation} \label{mphase1depth}
    O(\log k) + O(p^2 k) + O(\log k) + O(p\log k) = O(p^2 k)
\end{equation}
The number of communication gates across the $p$ QPUs is $O(p \log k)$. 

\textbf{Part 3.}

Now we combine this with Phase 2 and 3 within each QPU, as described in section \ref{sec:2QPU}. Summing the contributions from equations \ref{mphase1size}, \ref{phase2size}, and \ref{phase3size}, the total circuit size is
\begin{equation*}
    O(pk^2) + O(k) + O(nk) = O(nk).
\end{equation*}
Similarly, combining the bounds from equations \ref{mphase1depth}, \ref{phase2depth}, and \ref{phase3depth}, the total circuit depth satisfies
\begin{equation*}
    O(p^2 k) + O(\log^2 k) + O(k + \log k \log(n/k)) = O(p^2 k + \log k \log(n/k)),
\end{equation*}
The equation follows from our replacement of $k'$ with $k$. The communication complexity is $O(p \log k)$ as analyzed in Part 2. This completes the proof of Theorem \ref{pQPUthm}.
\end{proof}

\section{Lower Bound of Communication Complexity} \label{sec:lower}

In this section, we investigate lower bounds on communication complexity. Specifically, we show that preparing a state $\ket{\psi}$ across $p$ QPUs requires a communication complexity that is lower bounded by the logarithm of the CP-rank of a $p$-order tensor. This tensor is obtained from the state vector of $\ket{\psi}$ by reshaping it according to the qubit assignment among the QPUs. Unfortunately, computing the CP-rank of a general $p$-order tensor is NP-hard for $p \ge 3$\cite{haastad1990tensor}, and we are therefore unable to evaluate the CP-rank of the tensor associated with the Dicke state in this regime. However, for the case $p = 2$, the CP-rank reduces to the matrix rank. In this setting, we explicitly compute the rank of the Dicke state’s state matrix and show that our construction for $p = 2$ achieves this lower bound.

\subsection{$p$-QPU Distributed Preparation for General States}

We first introduce the state tensor of a quantum state.

\begin{definition}[State Tensor] \label{statetensor}
    Let $f: [n] \to [p]$ be a partition that assigns the $n$ qubits to $p$ QPUs, and let $Q_j = f^{-1}(j)$ denote the set of qubits assigned to $\QPU_j$. Any quantum state $\ket{\psi}$ distributed across the $p$ QPUs can be written as
    \begin{equation*}
        \ket{\psi} = \sum_{\substack{i_j \in \left[ 2^{|Q_j|} \right] \\ j \in [p]}} T^f_{i_0 i_1 \dots i_{p-1}} \ket{i_0 i_1 \dots i_{p-1}},
    \end{equation*}
    where $i_j$ indexes the computational basis states of $\QPU_j$ for each $j \in [p]$. Collecting these coefficients, we define the tensor
    \begin{equation*}
        T^f_{\ket{\psi}} = \left(T^f_{i_0 i_1 \dots i_{p-1}}\right)
    \end{equation*}
    to be the \emph{state tensor} associated with the state $\ket{\psi}$ and the partition $f$.
\end{definition}

Intuitively, a local unitary operation acting on $\QPU_j$ corresponds to a linear transformation applied along mode $j$ of the state tensor. As we show below, the communication complexity required to prepare $\ket{\psi}$ is closely related to the CP-rank of its associated state tensor.

\begin{theorem} \label{pQPUlowerbound}
    For any quantum state $\ket{\psi}$ and the number of QPUs $p$, the communication complexity satisfies
    \begin{equation*}
        \cc_p(\ket{\psi}) \ge \min_{\mathrm{balanced} f:[n]\to [p]} \log \CPrank\left(T^f_{\ket{\psi}}\right).
    \end{equation*}
\end{theorem}

\begin{proof}
    We prove the theorem via its contrapositive. Specifically, we show that for any quantum state $\ket{\psi}$ and the number of QPUs $p$, if the communication complexity of $\ket{\psi}$ on $p$ QPUs, $\cc_p(\ket{\psi}) = l$ for some integer $l \ge 0$, then for all balanced partition $f: [n]\to [p]$, the CP-rank of the corresponding state tensor satisfies $\CPrank\left(T^f_{\ket{\psi}}\right) \le 2^l$.
    
    We proceed by induction on the communication complexity $l$ for any balanced partition $f: [n]\to [p]$. For the base case $l = 0$, no communication is allowed, so the QPUs can only prepare a product state $\ket{\psi_0}$, where $\CPrank\left(T^f_{\ket{\psi_0}}\right) = 1$.

    Next, we consider the induction step. Suppose, by the induction hypothesis, that for communication complexity $l - 1$ we obtain
    \begin{equation*}
        \ket{\psi_{l - 1}} = \sum_{w \in [2]^{l - 1}} \ket{v_0(w), v_1(w), \dots, v_{p - 1}(w)}.
    \end{equation*}
    Without loss of generality, assume that the protocol requires $\QPU_0$ to perform a local operation and then communicate with $\QPU_1$. Let
    \begin{equation*}
        \ket{\psi'_{l - 1}} = \sum_{w \in [2]^{l - 1}} \ket{v'_0(w), v_1(w), \dots, v_{p - 1}(w)}
    \end{equation*}
    denote the state after the local operation on $\QPU_0$. When $\QPU_0$ sends a qubit to $\QPU_1$, it must first separate this qubit from its local system. In general, this qubit may be entangled with the other qubits in $\QPU_0$, so the resulting state can be written as
    \begin{equation*}
        \ket{\psi_l} = \ket{\psi'_{l - 1}} = \sum_{w \in [2]^{l - 1},\, z \in [2]} \ket{v''_0(w, z), c(w, z), v_1(w), \dots, v_{p - 1}(w)},
    \end{equation*}
    where $c(w, z)$ represents the communicated qubit and
    \begin{equation*}
        \ket{v'_0(w)} = \sum_{z \in [2]} \ket{v''_0(w, z), c(w, z)}.
    \end{equation*}
    It follows that $\CPrank\left(T^f_{\ket{\psi_l}}\right) \le 2^l$, as required.
\end{proof}

\subsection{$2$-QPU Distributed Preparation for Dicke States}

We establish a lower bound on the communication complexity required to prepare Dicke states across $2$ QPUs as follows.

\begin{theorem} \label{2QPUlowerboundDicke}
    To prepare $D(n,k)$ across two quantum processing units (QPUs), each holding $n/2$ qubits, the communication complexity satisfies 
    \begin{equation*}
        \cc_2(D(n, k)) \ge \lceil\log(k + 1)\rceil.
    \end{equation*}
\end{theorem}

Recall Theorem \ref{2QPUthm}. By Theorem \ref{2QPUlowerboundDicke}, we show that when $p = 2$, our construction for preparing Dicke states attains the lower bound on communication complexity.

The proof of Theorem \ref{2QPUlowerboundDicke} relies directly on Theorem \ref{pQPUlowerbound}. In the special case where $p = 2$, the CP-rank of the state tensor coincides with the rank of the state matrix, which is considerably easier to compute.

\begin{proof}
    By definition, $D(n, k)$ is symmetric under permutations of qubits. Consequently, it is unnecessary to specify the partition $f$ in the proof. Therefore, let $M_{D(n,k)} = (a_{x,y})$ be the matrix defined in Definition \ref{statetensor} when $p = 2$, and let $R_x$ denote the row of $M_{D(n,k)}$ indexed by $x$. By definition,
    \begin{equation*}
        a_{x,y} = \begin{cases}
        1, & \text{if } \ham(x) + \ham(y) = k, \\
        0, & \text{otherwise}.
        \end{cases}
    \end{equation*}
    
    We first observe that each row $R_x$ depends only on the Hamming weight $\ham(x)$. Indeed, if $\ham(x_1) = \ham(x_2)$, then for every $y$,
    \begin{equation*}
        \ham(x_1) + \ham(y) = k
        \Longleftrightarrow
        \ham(x_2) + \ham(y) = k,
    \end{equation*}
    and hence
    \begin{equation*}
        R_{x_1} = R_{x_2}.
    \end{equation*}
    Since $\ham(x)$ ranges over $\{0,1,\dots,k\}$ for rows that are not identically zero, there are at most $k+1$ distinct rows. Therefore,
    \begin{equation*}
        \rank(M_{D(n,k)}) \le k + 1.
    \end{equation*}
    
    On the other hand, if $\ham(x_1) \neq \ham(x_2)$, then the supports of the corresponding rows are disjoint. Indeed, there is no $y$ such that
    \begin{equation*}
        \ham(x_1) + \ham(y) = k
        \text{and}
        \ham(x_2) + \ham(y) = k
    \end{equation*}
    hold simultaneously. Consequently,
    \begin{equation*}
        \supp(R_{x_1}) \cap \supp(R_{x_2}) = \varnothing.
    \end{equation*}
    It follows that the rows corresponding to distinct Hamming weights are linearly independent, and thus
    \begin{equation*}
        \rank(M_{D(n,k)}) \ge k + 1.
    \end{equation*}
    
    Combining the above inequalities, we conclude that
    \begin{equation*}
        \rank(M_{D(n,k)}) = k + 1.
    \end{equation*}
    
    Substituting this into Theorem \ref{pQPUlowerbound}, we obtain
    \begin{equation*}
        \cc_2(D(n,k))
        \ge \left\lceil \log \rank(M_{D(n,k)}) \right\rceil
        = \left\lceil \log (k + 1) \right\rceil.
    \end{equation*}
\end{proof}

\section{Conclusion} \label{sec:conclusion}

This work addresses the critical scalability challenge in preparing large-qubit $k$-excitation Dicke states $D(n,k)$ by exploring their distributed preparation across a general number $p$ of QPUs. We present a novel distributed quantum circuit architecture that achieves both logarithmic communication complexity and polynomial circuit depth and size for distributed Dicke state preparation. Additionally, we establish a general lower bound on the communication complexity of $p$-QPU distributed state preparation, formulated via the CP-rank of a tensor associated with the target state. For $p = 2$, we explicitly compute the CP-rank of the corresponding Dicke state tensor, derive a lower bound of $\log k$, and verify that our construction’s communication complexity matches this fundamental limit, confirming its optimality in this case.

While this work advances the state of the art in distributed quantum state preparation, it also opens avenues for future research. We conjecture that our upper bound on communication complexity is tight for $p \geq 3$, though we have not yet proven this assertion. A key direction for future work is to derive the CP-rank of the Dicke state-associated tensor for $p \geq 3$---a result that would enable rigorous verification of our conjecture, establish general optimality bounds for distributed Dicke state preparation, and deepen the understanding of tensor ranks in the context of multi-QPU distributed quantum systems. Further extensions could also explore the robustness of our circuit to noise in inter-QPU communication channels, as well as its practical implementation on emerging distributed quantum hardware.


\balance
\bibliographystyle{IEEEtran} 
\bibliography{ref}

\end{document}